\documentclass{article}
\usepackage{spconf}

\usepackage[pdftex]{graphicx}
   \graphicspath{{figures/}}
   \DeclareGraphicsExtensions{.pdf,.jpg,.png,.eps}
\usepackage{epstopdf}
\usepackage{amsmath}
\usepackage{amsthm}
\usepackage{amssymb}
\usepackage{enumerate}
\usepackage{booktabs}
\usepackage{tabulary}
\usepackage{caption}
\usepackage{subcaption}
\usepackage[colorlinks=true,citecolor=black,urlcolor=blue,linkcolor=black]{hyperref}
\usepackage{cite}
\usepackage{notoccite}
\usepackage{mathtools}
\usepackage{pgfplots}
\pgfplotsset{compat=newest}
\usetikzlibrary{plotmarks,calc}
\usepgfplotslibrary{groupplots}

\usepackage{color}
\usepackage{tikz}

\newtheorem{theorem}{Theorem}

\renewcommand{\ref}[1]{(\oldref{#1})}

\begin{document}

\title{Age of Information With Finite Horizon and Partial Updates}
\name{David~Ramirez$^{\star}$ \qquad Elza~Erkip$^{\dagger}$ \qquad H. Vincent Poor$^{\star}$\thanks{Work partially done while Ramirez was at New York University.}}

			\address{$^{\star}$Department of Electrical Engineering, Princeton University, 
			    $^{\dagger}$New York University}

\maketitle

\begin{abstract}
A resource-constrained system monitors a source of information by requesting a finite number of updates subject to random transmission delays. An a priori fixed update request policy is shown to minimize a polynomial penalty function of the age of information over arbitrary time horizons. Partial updates, compressed updates with reduced transmission and information content, in the presented model are shown to incur an age penalty independent of the compression. Finite horizons are shown to have better performance in terms of second order statistic relative to infinite horizons.
\end{abstract}


\section{Introduction}
Tactile internet is on the horizon for wireless network services, if highly responsive connectivity between environment-facing sensors and user-facing systems is available \cite{simsek20165g}.
Sensor updates must be sparse enough to extend battery life yet frequent enough to satisfy a given system performance requirement. A pressing question, and the focus of this work, is to find an update request policy that gracefully maintains system performance of a resource constrained system.

When actions depend on the most recently sensed data, system performance degrades as the elapsed time since the last received update grows. The concept of age of information (AoI), as presented by Kaul et al. \cite{kaul2011minimizing}, quantifies the system performance penalty associated with the delay between received updates. AoI is a meaningful metric which differs from delay and throughput by focusing on the time between update generation and update arrival \cite{kadota2016minimizing}, i.e., AoI is a receiver-centric measure of delay. Achieving minimal AoI can be nontrivial, e.g., \cite{sun2017update} showed constraining the number of updates makes a zero-wait policy sub-optimal in terms of AoI. 

Results in \cite{chen2016age} characterized the impact of various queuing strategies on the peak AoI. In \cite{ceran2017average} it was shown that an AoI threshold strategy minimizes the average AoI under long-term average update constraints and error-correction techniques. Work in \cite{hsu2017age} showed that with adequate queuing strategies the impact of buffering capacity of intermediary relays becomes nearly negligible. Non-linear AoI penalties for a randomly update generating source were considered in \cite{kosta2017age}. A coding-inspired approach where updates vary in information quantity and bit size was presented in \cite{poojary2017real}. A system monitoring a random source with a finite number of updates was considered in \cite{chakravorty2015distortion} where a threshold based strategy optimally minimized the expected estimation error. A controllable update processing time was considered in \cite{bastopcu2019age} to minimize AoI constrained to a distortion metric inversely proportional to the processing time of updates sent over a deterministic channel. Mutual information (MI) of a source and destination was proposed as a measure for AoI in \cite{sun2018information}, where an MI thresholding strategy was shown to be optimal. Recently, \cite{ayan2019age} considered a central scheduler tasked with minimizing the total AoI by controlling how updates are relayed from sources to destinations.

The present work is distinguished via the use of a model that allows for new update requests to be made even if a previous update has not yet arrived, operates over arbitrary time horizons, and considers the use of partial updates. The AoI penalty is a polynomial function of arbitrary degree and finite first and second moments are assumed computable for the random delay. Minimization of the total AoI penalty is obtained by a priori computing update request times. The impact on AoI penalty of reordering update requests is presented. Additionally, it is shown that a finite horizon approach achieves smaller variance of AoI penalty, thus greater stability can be provided to systems with high reliability requirements.

\section{Network Model}
A monitoring system requests information updates via wireless links from a source of information (e.g., a set of sensors). The system is limited to making at most $N\in \mathbb{Z}^+$ update requests over a time period $T\in \mathbb{R}^+$. Note that $\frac{N}{T}$ can serve as an average update constraint or a constraint on energy if an energy per update ratio is introduced.

Define $i \in \{1,...,N\}$ as a labeling of the $N$ updates, $t:0\leq t \leq T, t \in \mathbb{R}$ as a time index, $\delta_i \in [0,T]$ as the departure time of the $i$-th update request. \emph{W.l.o.g.}, labels are such that $\delta_i \leq \delta_j~\forall~ i<j \in \{1,...,N\}$. Define $d_i$ as the transmission time of update $i$, modeled as a continuous strictly positive random variable with finite first and second moments. Statistics of all transmission times are assumed independent from each other, but not identical to allow the model to account for the $N$ updates to arrive from potentially $N$ distinct sensors. Define $\alpha_i = \delta_i + d_i$ as the arrival time of update $i$.

To account for, either, multiple paths between system and sensors or higher layer mechanisms giving sequential requests different delays (e.g., exponential backoff in $802.11$), assume new update requests can be made even if a previous update has not yet arrived. At time $t$, define the \emph{freshest} update $u(t)$ as $\arg \max_{i} \delta_i:\alpha_i\leq t$ with $u(t)=0 ~\forall~t<\alpha_1$, i.e., the freshest update is a received update that was generated most recently. Update $i$ is said to be useful if $\exists~t : i=u(t)$ else the update is wasted. 

Define $\boldsymbol{\delta}=\{\delta_i~\forall~i\}$ as a policy, and note that the system only controls what policy is used. Policy selection impacts the system penalty function, defined as a function which is discontinuous at reception of useful updates and otherwise monotonically increasing with time.
Related literature names this penalty function the \emph{age of information} and is defined as
\begin{equation}
A_k(t, \boldsymbol{\delta}) = C_k (t - \alpha_{u(t)} + d_{u(t)})^k 
\end{equation}
\noindent with $k \geq 1$, $C_k \in \mathbb{R}^+$ as a scaling value, $A_0 \in \mathbb{R}^+$ as an initial age, and $d_0 = \left(\frac{A_0}{C_k}\right)^{1/k}$. Define the set of useful updates $\boldsymbol{\alpha}^u = \{\alpha_j~\forall~j:~\exists~t,u(t)=j\} \cup \{a^u_0=0, a^u_{N+1}=T \}$ and the total penalty can be defined as 

\begin{equation}
\mathcal{S}_k(\boldsymbol{\delta})= \sum_{i=0}^{N} \int_{\alpha^u_i}^{\alpha^u_{i+1}} A_k (t, \boldsymbol{\delta}) dt,
\label{totalpenalty}
\end{equation}
\noindent where $A_k(t, \boldsymbol{\delta})$ is evaluated from inside the integration period. Note that $\mathcal{S}_k(\boldsymbol{\delta})$ is a random variable.

If a policy $\boldsymbol{\delta}$ achieves $\mathcal{S}_k(\boldsymbol{\delta})$ with $N$ updates, then there exists a policy with $N+1$ updates with no larger a total penalty (e.g., the same $N$ departure times and random departure time of the extra update). Furthermore, if a policy with $N+1$ \emph{useful} updates exists, then a lower total penalty can be guaranteed. The problem of interest is to obtain an optimal policy, i.e.,

\begin{equation}
\boldsymbol{\delta}^* =  \arg\min~E[\mathcal{S}_k(\boldsymbol{\delta})].
\label{OptProbL}
\end{equation}
\noindent Note, an optimal policy has a penalty no greater than any other policy when averaged over \emph{many} instances of length $T$. 

\section{Optimal Solution: Critical Age Policy}
First, a set of feasible solutions for (\ref{OptProbL}) for a given limited expected peak age is defined. Then, a proof for the existence of an optimal solution and method to find an optimal policy are presented.

To account for discontinuities due to update reception in ($1$), define $f^-(t)=\lim\limits_{t'\to t^-}f(t')$ and $f^+(t)=\lim\limits_{t'\to t^+}f(t')$ for some function $f(t)$. Define an arbitrary limit on the expected peak age as the critical age $A^C$. For a critical age $A^C > A^-_k(E[d_1],\boldsymbol{\delta})$, the critical policy is defined as 
\begin{equation}
\boldsymbol{\delta}^C=\{\delta_i^C : E[A^-_k(t,\boldsymbol{\delta}^C)] \leq A^C~\forall~t\} ,
\end{equation}
\noindent else if $A^C \leq A^-_k(E[d_1],\boldsymbol{\delta})$, the critical policy $\boldsymbol{\delta}^C$ is 
\begin{equation}
\boldsymbol{\delta}^C=\{\delta_1^C=0,~
\delta_i^C : E[A^-_k(t,\boldsymbol{\delta}^C)] \leq A^C~\forall~t > \alpha_1\},
\label{2daDef}
\end{equation}
\noindent where (\ref{2daDef}) sets $\delta_1^C=0$ to mitigate $A_0$ as quickly as possible since $A_k$ is smaller than the age at the earliest expected arrival of the first update (i.e., if $\delta_1=0$ then $E[\alpha_1]=E[d_1]$). 

For a given $N$, define the set of all critical policies as $\mathcal{P}_N$. A critical age of $0< A^C < \max_i E[d_i]$ cannot be guaranteed by any critical policy. Thus, among all policies in $\mathcal{P}_N$ there exists a minimal $A^C$. Define $\boldsymbol{\delta}^* := \arg \min_{\boldsymbol{\delta} \in \mathcal{P}_N} A^C$ as the critical policy with smallest critical age $A^*:=\min_{\boldsymbol{\delta} \in \mathcal{P}_N} A^C$. Theorem $1$ asserts that a critical policy solves (\ref{OptProbL}).

\begin{theorem}
A critical policy $\boldsymbol{\delta}^* \in \mathcal{P}_N$ obtains an average total penalty no greater than any other policy.
\end{theorem}
\begin{proof}
Define a policy that differs from $\boldsymbol{\delta}^*$ in at least one element, i.e., $\boldsymbol{\delta}':\exists~ \delta'_i \neq \delta^*_i$. Furthermore, define a subset of all the differing elements as $\boldsymbol{M}=\{m:\delta'_m \neq \delta^*_m, 1 \leq m \leq N\}$. Assume that $\boldsymbol{\delta}'$ and $\boldsymbol{\delta}^*$ have equal numbers of useful updates, else $\boldsymbol{\delta}'$ can be improved upon. Define $\Delta_m = (\delta'_m - \delta'_{m-1}) - (\delta^*_m - \delta^*_{m-1})~\forall~m \in \boldsymbol{M}$, with $\delta'_0=\delta^*_0=0$ for completeness, as the temporal shift between the policies for consecutive departures of updates in the set $\boldsymbol{M}$. 

Delay, i.e. $d_i~\forall~i$, is a policy independent stochastic process, thus $A_k^+(\alpha^*_m,\boldsymbol{\delta}^*)=A_k^+(\alpha'_m,\boldsymbol{\delta}')$ and the penalty incurred by $\boldsymbol{\delta}'$ for $t \in [\alpha'_m, \min(\alpha^*_{m+1}-\Delta_m, \alpha'_{m+1})]$ is equally incurred by policy $\boldsymbol{\delta}^*$ for $t \in [\alpha^*_m, \min(\alpha^*_{m+1}, \alpha'_{m+1}-\Delta_m)]$. 

Focus on the penalty incurred during the remaining time period of length $|\Delta_m|$. First, if $|\Delta_m|=0$ then the impact of update $m$ for policy $\boldsymbol{\delta}'$ is equivalent and time shifted version of the impact of update $m$ for policy $\boldsymbol{\delta}^*$. Note that $\exists~m : |\Delta_m|\neq 0$ since both policies have an equal and finite number of useful updates and $T$ is finite. Note that the time period of length $|\Delta_m|$ need not be a continuous subset of the time period $[0,T]$. 

Second, if $\delta^*_1 \neq 0$, then over \emph{any} time period of length $|\Delta_m|$ the policy $\boldsymbol{\delta}^*$ incurs a total penalty no greater than $|\Delta_m|A^C$. Alternatively, if $\delta_1^*=0$ the total penalty incurred by both policies over time period $[0,\alpha_1]$ is equivalent, and for $\boldsymbol{\delta}^*$ no greater than $|\Delta_m| A^C$ for any period of length $|\Delta_m|$ in the interval $[\alpha_1,T]$.

Third, for any $\Delta_m>0$ (i.e., transmitting later than $\boldsymbol{\delta}^*$) policy $\boldsymbol{\delta}'$ incurs a total penalty no less than $|\Delta_m|A^C$ at a time period after $\alpha^*_m$ since the penalty function is monotonically increasing. For any $\Delta_m<0$ (i.e., transmitting earlier than $\boldsymbol{\delta}^*$) policy $\boldsymbol{\delta}'$ incurs a penalty after $\alpha'_m$ equivalent to the penalty incurred by $\boldsymbol{\delta}^*$ after $\alpha^*_m$; yet since $T$  and $N$ are finite there exists a period of length $|\Delta_m|$ for which $\boldsymbol{\delta}'$ incurs a penalty greater than $|\Delta_m|A^C$. 
\end{proof}

Essentially, any deviation from the optimal policy $\boldsymbol{\delta}^*$ creates a temporal shift of the penalty function. Since the penalty function is non-negative and positive monotonic, a larger penalty is incurred. Fig. \ref{FigThm} illustrates possible values of $\Delta_m$ as described in the previous proof. 

\begin{figure}[ht]
\centering
\includegraphics[width=0.85\columnwidth]{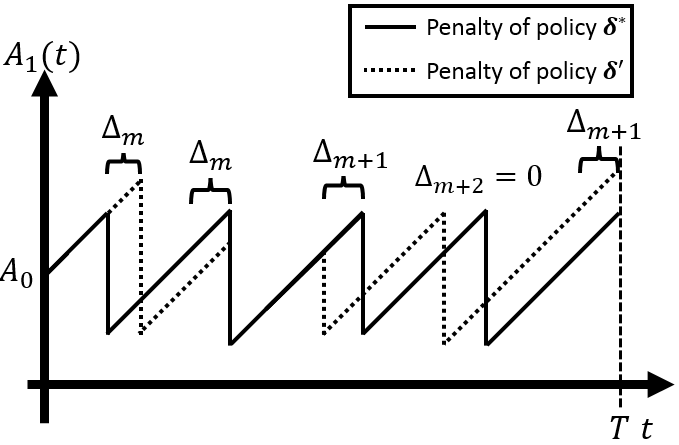}
\caption{Example of possible differences between $\boldsymbol{\delta}^*$ and $\boldsymbol{\delta}'$. Updates $m$ and $m+1$ cause a greater penalty in $\boldsymbol{\delta}'$, while update $m+2$ does not by itself cause a penalty difference.}

\label{FigThm}
\end{figure}

Construction of $\boldsymbol{\delta}^*$ provides one equation for each of the $N$ departure times $\delta_i$ with $N+1$ variables, i.e., the $N$ departure times plus the critical age $A^C$. The additional equation necessary to make the system solvable comes from realizing that at time $T$ the age should also not exceed $A^C$. Thus a critical policy $\boldsymbol{\delta}^*$ satisfies the following:

\begin{equation}
A_k^-(\delta_i^* + E[d_i], \boldsymbol{\delta}^*) = A^*~\forall~i \text{ and } A_k^-(T, \boldsymbol{\delta}^*)=A^*,
\end{equation}

\noindent 
from which one can obtain the value of
\begin{equation}
A^* = C_k \left(\frac{A_0 +T+\sum_{i=1}^N E[d_i]}{N+1}\right)^{k}.
\label{EqAC}
\end{equation}
\noindent Furthermore, one can obtain 
\begin{equation}
\delta_N^* = \frac{NT - A_0 - \sum_{i=1}^{N}E[d_i]}{N+1}
\label{eqLastN}
\end{equation}
\noindent and then recursively solve for the remaining $N-1$ values. Finally, the average total penalty incurred by policy $\boldsymbol{\delta}^*$ is \begin{equation}
\overline{\mathcal{S}_k}(\boldsymbol{\delta^*})=\sum_{i=0}^{N} \frac{C_k\left[\left(\delta^*_{i+1} + E[d_{i+1}] - \delta^*_{i} \right)^{k+1}-E[d_i]^{k+1}\right]}{(k+1)} ,
\label{AvgLinPen}
 \end{equation}\noindent with $\delta^*_0=d_0=-A_0^{1/k}, d_{N+1}=0,$ and $\delta^*_{N+1}=T$ defined for completeness and compactness.

For large values of $A_0$ or $E[d_i]$, e.g., $A_0 > NT - \sum_{i=1}^N E[d_i]$, optimal departure times may become negative, cf. (\ref{eqLastN}). In such cases, the first update is $\delta_1=0$ and the problem is recast with the remaining $N-1$ departures. While optimal departure times for any value of $k$ are equivalent, the critical age of the policy is distinct.

Augmenting a policy with additional \emph{useful} updates decreases the total age penalty. Yet, reordering update requests can turn a useful update into a wasteful update. In the simple case when all $N$ update requests experience statistically equivalent delays, the time between any two consecutive departures is constant. Formally, if  $E[d_i]=E[d_j]~\forall~i,j$ then $\delta^*_i - \delta^*_{i-1} = \delta^*_j - \delta^*_{j-1}~\forall~i,j$. Meaning, if all updates are transmitted over the same stationary wireless channel, then a constant time period between update requests is optimal and reordering updates does not impact the solution of (\ref{OptProbL}). 

Now, consider only the case when $\exists~i,j:E[d_i] \neq E[d_j]$. First, if $k=1$ and $A_0=0$ then any achievable value of $A^C$ by a policy $\boldsymbol{\delta}^*$ is achievable by any reordering of the update requests with a different policy. This follows from (\ref{EqAC}) being independent of $\boldsymbol{\delta}^*$ and $A_0=0$. Second, if $A_0 >> 0$ the first departure may be forcefully fixed to $\delta_1=0$ to mitigate the penalty as soon as possible. Intuitively, when reordering is possible and a value of $A_0$ is large enough that $\delta^*_1=0$, then the first update should be requested from the source with minimum delay, i.e., reorder such that $\arg \min_i E[d_i]=1$. Third, if only the first element is forcefully fixed due to a large initial penalty, then the remaining $N-1$ elements can achieve equivalent average total age penalty with any reordering. Reordering, as described, does not impact $A^*$ but does impact the values of $\delta_i \in \boldsymbol{\delta}^*$.

\section{Partial Updates}
In this section, and only this section, consider that an update can be compressed to reduce transmission time and information content (i.e., lossy compression). The compressed update is a partial update, i.e., for update $i$ with transmission time $d_i$ has a corresponding partial update with transmission time $d_p$ such that $E[d_p]<E[d_i]$ and $A_k^+(\alpha_i,\boldsymbol{\delta})<A_k^+(\alpha_p,\boldsymbol{\delta})$ if $A_k^-(\alpha_i,\boldsymbol{\delta})=A_k^-(\alpha_p,\boldsymbol{\delta})$. Compare with \cite{poojary2017real} where a partial update was a lossless compressed version of a full update. 

Assume a partial update \emph{resets} the AoI to the AoI at reception of the previous useful update. Meaning, for a partial update $p \in \{1,...,N\}$ at time $t$ such that $p=u(t)$ define the previous update arrival time as $t'=\lim_{\epsilon \to 0} \alpha_{u(t-\epsilon)}$ with $\alpha_0=0$, and $A_k^+(\alpha_p,\boldsymbol{\delta}):=A_k^+(t',\boldsymbol{\delta})$. Thus, the partial update provides no \emph{fresher} information than that of the previous update. Such a model for partial update inherently favors the use of partial updates over a full update whenever the expected delay increases the AoI penalty no greater than the penalty at the previous update arrival.

An optimal critical policy (cf. Section III) using only partial updates, i.e., $\boldsymbol{\delta}_p$, can lead to a total penalty that is independent of transmission times. E.g., with $k=1$ the optimal critical policy creates a sawtooth wave of $N+1$ peaks that go from $A_0$ up to $A_0+C_k\frac{T}{N+1}$ thus obtaining a total penalty of
\begin{equation}
S_k(\boldsymbol{\delta}_p)=\frac{C_kT^2}{2(N+1)}+A_0 T.
\label{SumPartialN}
\end{equation}
\noindent Note the lack of dependence of (\ref{SumPartialN}) on the update delay, meaning that reduced delay of a partial update does not impact the performance. When the delay of full updates is $d_i \leq \frac{A_o}{C}~\forall~i$, then a policy of only full updates will not incur a larger total linear penalty than a policy of only partial updates. 

\section{Impact on Second Order Statistics}
Fixing $N$ and increasing $T$ leads to an increase in (\ref{OptProbL}), while fixing $T$ and increasing $N$ can decrease (\ref{OptProbL}). Fixing the ratio $\frac{N}{T}$ allows for arbitrary, yet equal, growth in $N$ and $T$ which maintains equivalent average resource consumption. From (\ref{EqAC}) one can infer that increasing $N$ and $T$ at the same rate will drive $A^*$ towards $C_k \left(\frac{T}{N}+\overline{E[d_i]} \right)^{1/k}$ where $\overline{E[d_i]}$ is the arithmetic mean of expected delays.

\begin{figure}
\centering
  \includegraphics[width=.8\linewidth]{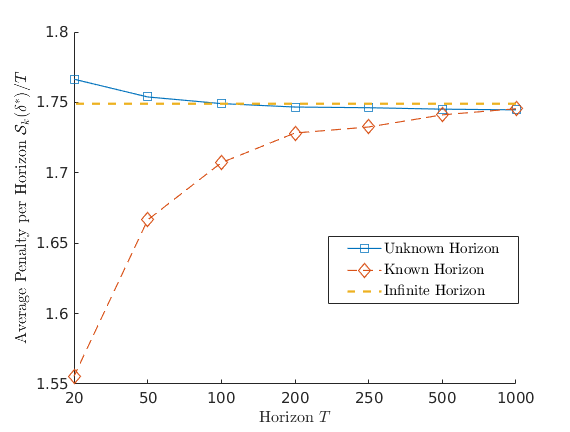}
  \caption{Avg. Penalty per Horizon $\mathcal{S}_k(\boldsymbol{\delta}^*)/T$ as a function of $T$. t $T>100$, the Unknown Horizon outpreforms the Infinite Horizon, and at $T>250$, the scenarios are comparable.}
\label{CompAvg}
\end{figure}

From (\ref{AvgLinPen}), note that increasing $N$ and $T$ by equal amounts leads to an increase in the total penalty $\mathcal{S}_k(\boldsymbol{\delta}^*)$. Adding to the growth in average penalty is the negative impact on the variance of the penalty experienced over many realizations, defined as $\sigma_{\mathcal{S}_k(\boldsymbol{\delta}^*)}^2$. With (\ref{AvgLinPen}) upper and lower bounded by $TA^*$ and $\min (A_0, E[d_i]~\forall~i)$, respectively, one can use Bhatia-Davis' inequality \cite{bhatia2000better} to upper bound the variance as
\begin{equation}
\sigma_{\mathcal{S}_k(\boldsymbol{\delta}^*)}^2 < \left(TA^*-\overline{\mathcal{S}_k}(\boldsymbol{\delta^*})\right)\left(\overline{\mathcal{S}_k}(\boldsymbol{\delta^*}) -\min (A_0, E[d_i]~\forall~i) \right)
\label{secondord}
\end{equation}
\noindent where various refinements can be done if additional assumptions are made \cite{agarwal2005survey}. Now, if increasing $N$ and $T$ is done such that the additional update requests do not reduce the value of $\min (A_0, E[di]~\forall~i)$, then one can expect the variance of the total penalty incurred to increase with $N$ and $T$. Therefore, a shorter horizon (i.e., smaller $T$) can lead to smaller values of $\sigma^2$, and ultimately more stable performance from the system.

\begin{figure}
\centering
\includegraphics[width=.8\linewidth]{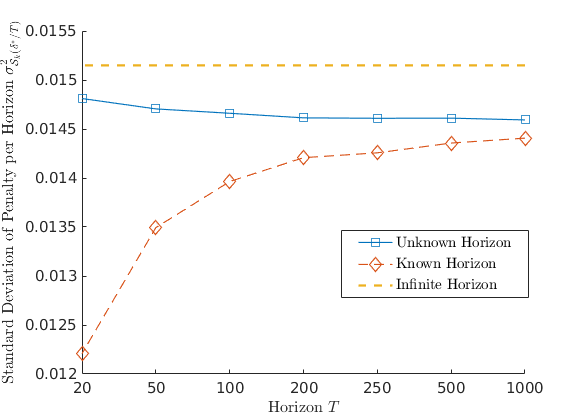}
  \caption{Standard deviation of $\sigma_{\mathcal{S}_k(\boldsymbol{\delta}^*)}^2$ as a function of $T$. Infinite Horizon underperforms among evaluated scenarios.} 
\label{CompVar}
\end{figure}

The elapsed time before the system is set to \emph{initial conditions} defines the horizon $T$. Defining \emph{initial conditions} is where system engineering comes into play. A system where delay is temporally stationary has a horizon $T=\infty$, while a system where the estimated statistics on delay are applicable for a finite amount of time before they need to be re-estimated has a horizon $T<\infty$. Horizon selection impacts second order statistics, as suggested by (\ref{secondord}) and analyzed in the sequel. 

\section{Simulation Results}
Simulation parameters are set to $\frac{N}{T}=0.4, A_0 = 0, k=1, C_k=1,$ and $d_i \sim U(0,1)$. Three scenarios are evaluated. Infinite Horizon solves (\ref{OptProbL}) with $T=\infty$. Known Horizon uses $T$ as the horizon. Unknown Horizon concatenates periods of length $\frac{1000}{T}$, using $\frac{N}{T}$ updates per period, and the penalty at the end of one period is the initial age of the next period. 

Fig. \ref{CompAvg} shows $E[\frac{\mathcal{S}_k(\boldsymbol{\delta}^*)}{T}]$ as a function of $T$, where a lower value on the vertical axis implies better performance. Intuitively, using larger $T$ moves all scenarios towards the Infinite Horizon scenario. Note the existence of a trade-off point ($T < 100$) where assuming an infinite horizon is a better strategy than estimating an unknown horizon. The Unknown Horizon scenario suffers from the initial age carried over from previous periods, thus its performance improves with increase in $T$. When the horizon is known, a smaller horizon allows for a finer control of resources, thus a smaller penalty. In Fig. \ref{CompVar} the standard deviation of $\frac{\mathcal{S}_k(\boldsymbol{\delta}^*)}{T}$ is shown as a function of $T$, where a lower value implies better performance. Assuming an infinite horizon leads to a larger standard deviation, and a greater variation in performance.

\section{Conclusion}

An update request policy that minimizes the average information age penalty of a resource constrained monitoring system with stochastic update delay has been presented. The optimal departure times have been obtained by solving a linear system of equations which grows with the available number of updates. Results have also outlined that using partial updates incurs a total age independent of delay, but is only beneficial when the delay reduction outweighs the increase in age. 

\bibliographystyle{IEEEtran}
\bibliography{AoI.bib}

\end{document}